\documentclass[11pt,leqno]{amsart}
\usepackage{amsmath}
\usepackage{amssymb}
\usepackage{a4wide}
\usepackage{bm}
\usepackage{graphics}
\usepackage{hyperref}
\usepackage{color}
\usepackage{braket}
\definecolor{ao(english)}{rgb}{0.0, 0.5, 0.0}
\hypersetup{%
  colorlinks = true,
  linkcolor  = blue,
  citecolor    = ao(english)
}
\usepackage{comment}
 \usepackage[applemac]{inputenc}
\usepackage{latexsym, amscd, amsthm,amsfonts,amstext}
\usepackage[mathscr]{eucal}
\usepackage{appendix}
\usepackage[active]{srcltx}
\include{srctex}

 \usepackage{mleftright}
\usepackage{pgf,tikz}
\usetikzlibrary{patterns,arrows,intersections, decorations.markings, decorations.pathmorphing,
	backgrounds, positioning, fit, shapes.geometric}
\usepackage{caption}  
\usepackage{mathrsfs}
\usepackage{xparse}
\makeatletter

\usepackage{bm}

\usepackage{cancel}

\RenewDocumentCommand{\title}{om}{%
   \IfNoValueTF{#1}
     {\gdef\shorttitle{Transition of the semiclassical resonance widths}}%
     {\gdef\shorttitle{#1}}%
   \gdef\@title{#2}%
}

\DeclareFontFamily{U}{mathb}{\hyphenchar\font45}
\DeclareFontShape{U}{mathb}{m}{n}{ <-6> matha5 <6-7> matha6 <7-8>
mathb7 <8-9> mathb8 <9-10> mathb9 <10-12> mathb10 <12-> mathb12 }{}
\DeclareSymbolFont{mathb}{U}{mathb}{m}{n}

\DeclareMathAccent{\abxring}{0}{mathb}{"38}

\DeclareFontFamily{U}{mathb}{\hyphenchar\font45}
\DeclareFontShape{U}{mathb}{m}{n}{ <-6> matha5 <6-7> matha6 <7-8>
mathb7 <8-9> mathb8 <9-10> mathb9 <10-12> mathb10 <12-> mathb12 }{}
\DeclareSymbolFont{mathb}{U}{mathb}{m}{n}

\usepackage{subcaption} 

\newcommand{\N}{\mathbb{N}}

\newcommand{\R}{\mathbb{R}}
\newcommand{\C}{\mathbb{C}}
\newcommand{\W}{{\mathcal W}}

\newcommand{\cA}{{\mathcal A}}
\newcommand{\cR}{\mathcal{R}}

\newcommand{\Ai}{\ope{Ai}}
\newcommand{\Bi}{\ope{Bi}}
\newcommand{\Ci}{\ope{Ci}}

\newcommand{\p}{\partial}

\newcommand{\dl}{\delta}

\newcommand{\til}{\widetilde}

\newcommand{\re}{{\rm Re}\hskip 1pt }
\newcommand{\im}{{\rm Im}\hskip 1pt }
\newcommand{\ord}{{\mathcal O}}
\newcommand{\ai}{{\rm Ai}\,}

\newcommand{\ope}[1]{\operatorname{#1}}
\newcommand{\mc}[1]{\mathcal{#1}}

\newcommand{\qtext}[1]{\quad\text{#1 }\ }

\newcommand{\be}{\begin{equation}}
\newcommand{\ee}{\end{equation}}
\newcommand{\ben}{\begin{equation*}}
\newcommand{\een}{\end{equation*}}

\newcommand{\inc}{{
\flat}}
\newcommand{\out}{{
\sharp}}
\newcommand{\dir}{{\bullet}}

\newcommand{\boundellipse}[3]
{(#1) ellipse (#2 and #3)
}

\makeatletter
 \@addtoreset{equation}{section}
 \makeatother
 
 \usepackage{hyperref} 

\theoremstyle{theorem} 
\newtheorem{theorem}{Theorem}
\newtheorem{lemma}{Lemma}[section]

\newcounter{Cond}

\newtheorem{condition}[Cond]{Assumption}
\newtheorem{proposition}[lemma]{Proposition}
\newtheorem{remark}[lemma]{Remark}

\newcounter{Ass}

\newtheorem{As}{Assumption}[Ass]

\theoremstyle{definition} 

\numberwithin{equation}{section}
  {\setlength{\baselineskip}{1.5\baselineskip}

\title{
Transition of the semiclassical resonance widths  \\
across a tangential crossing energy-level 
}
\author{Marouane Assal}
\address{Marouane Assal, Departamento de Matem\'atica y Ciencia de la Computaci\'on, Universidad de Santiago de Chile, Las sophoras 173, Santiago, Chile. e-mail: marouane.assal@usach.cl} 

\author{Setsuro Fujiie}
\address{Setsuro Fujiie, Department of Mathematical Sciences, 
Ritsumeikan University, 111 Noji-Higashi, Kusatsu, 525-8577,  Japan. 
e-mail: fujiie@fc.ritsumei.ac.jp}
\author{Kenta Higuchi}
\address{Kenta Higuchi, Graduate School of Science and Engineering, Ehime University/ Bunkyocho 3, Matsuyama, Ehime, 790-8577, Japan.
 e-mail: higuchi.kenta.en@ehime-u.ac.jp}

\begin{document}

\maketitle

\begin{abstract}
We consider a 1D $2\times 2$ matrix-valued operator \eqref{System0} with two semiclassical Schr\"odinger operators on the diagonal entries and small interactions on the off-diagonal ones. When the two potentials cross at a turning point with
contact order $n$, the corresponding two classical trajectories at the crossing level  intersect at one  point in the phase space with contact order $2n$.
We compute the transfer matrix at this point between the incoming and outgoing microlocal solutions and apply it to the semiclassical distribution of resonances at the energy crossing level. It is described in terms of a generalized Airy function.
This result generalizes \cite{FMW1} to the tangential crossing and \cite{AFH1} to the crossing at a turning point.

\vspace{0.3cm}

\end{abstract}


\section{Introduction}
We study the asymptotic distribution of resonances for the matrix Schr\"odinger operator
\begin{equation}\label{System0} 
P := 
\begin{pmatrix}
P_1 & h U\\
h U^* & P_2
\end{pmatrix},
\end{equation} 
with a small positive parameter $h$. 
The diagonal part consists of scalar Schr\"odinger operators 
\ben
P_j:=(h D_x)^2  + V_j(x) \quad (j=1,2), \;\;\;\; D_x:= -i \frac d{dx},
\een
with real-valued potentials  $V_1$ and $V_2$ on $\mathbb R$, and the anti-diagonal part describing the interaction between $P_1$ and $P_2$ consists of ($h$ times) a first order semiclassical differential operator $U$ with real-valued coefficients on $\mathbb R$
\ben
U=U(x,hD_x) := r_0(x) + i  r_1(x) h D_x,
\een
and its adjoint $U^*$. Such operators arise as important models in molecular physics and quantum chemistry for instance in the Born-Oppenheimer approximation of molecular dynamics (see e.g., \cite{KMSW}). Our goal is to provide a precise asymptotics on the imaginary parts (widths) of the resonances which correspond to the inverse of the lifetime of the molecule.

Let $p_j(x,\xi):= \xi^2+ V_j(x)$ be the classical Hamiltonian associated with $P_j$ $(j=1,2)$.
We consider the situation where, at an energy-level $E_0$ (we set $E_0=0$), the classical trajectory on $\Gamma_1:=p_1^{-1}(0)$ is a simple closed curve, while that on $\Gamma_2:=p_2^{-1}(0)$ is non-trapping. In this situation, the scalar operator $P_1$ has eigenvalues near $E=0$ subject to the Bohr-Sommerfeld quantization rule \eqref{BSR}, but because of the
interaction with $P_2$, these eigenvalues may be shifted to resonances  of the matrix operator $P$ in the lower half plane (Fermi's golden rule). 

The semiclassical asymptotics of the imaginary parts of the resonances in this setting has first been studied  in \cite{Ma,Na,Ba,As} in the case where $\Gamma_1$ and $\Gamma_2$ are disjoint. They proved the exponential decay estimate  of the resonance widths.
Recently,  \cite{FMW1,FMW2,FMW3, AFH1} studied the case where $\Gamma_1$ and $\Gamma_2$ cross each other and proved the polynomial asymptotics of the resonance widths. 
More precisely, \cite{FMW1, FMW2}  considered the resonances near a crossing level using a construction of exact solutions to the system $Pu=Eu$.
On the other hand, \cite{FMW3, AFH1} studied the resonances above the crossing level
using a microlocal argument. 
The microlocal argument reduces the study of resonances to a microlocal transfer matrix at the crossing point, which relates the microlocal data on the two incoming trajectories to the two outgoing ones. 
The recent work \cite{AFH1} generalized the microlocal transfer matrix  at a transversal crossing point of  \cite{FMW3} to a tangential one, by making use of the exact solutions constructed in \cite{FMW1, FMW2}.

In the present article, we  study the resonances at the  crossing level, which was excluded in \cite{AFH1}, and generalize \cite{FMW1,FMW2} to the tangential crossing. 
A remarkable feature at the crossing level is the transition of the intersection between $\Gamma_1$ and $\Gamma_2$.
 Below this level, they have no intersection, which results in the exponentially small width of resonances. Above the crossing level, on the contrary, they intersect transversally at two points, which implies the width of polynomial order as proved in \cite{AFH1}.  
Our result Theorem \ref{MAINTH} asserts that the resonance widths are of order $h^{1+\frac 2{2n+1}}$, where $n$ is the contact order of  $V_1$ and $V_2$, and the coefficient of the leading term is explicitly given in terms of a generalization of the
 Airy function. Notice that the contact order of $\Gamma_1$ and $\Gamma_2$ is $m=2n$ and the order $h^{1+\frac 2{2n+1}}=h^{1+\frac 2{m+1}}$ of the resonance widths coincides with that `above' an order $m$ crossing obtained in \cite{AFH1}. The generalized Airy function $A_n(\lambda)$, which decays exponentially as $\lambda\to +\infty$ and oscillates as $\lambda\to -\infty$, describes the transition.
 
The proof of Theorem \ref{MAINTH} is reduced to the asymptotic  formula of the microlocal transfer matrix at the crossing point (Theorem \ref{prop:crossing-turning}). In particular, the exponent of the order in $h$ of the resonance width is one plus twice the exponent of the subprincipal order of the transfer matrix, which is $1/(2n+1)$ in our case.
This reduction argument is exactly the same as in \cite{AFH1}, and is omitted in this article.

Theorem \ref{prop:crossing-turning} is the linear relation between the incoming microlocal data and the outgoing one near a crossing point. Because of the tangential crossing, we cannot reduce it to a normal form. Instead, we construct local exact solutions by a convergent successive approximation. More precisely, let the crossing point be $(x,\xi)=(0,0)$. We construct a basis of exact solutions $(w_1^\flat, w_2^\flat, w_1^\sharp, w_2^\sharp)$ near $x=0$ and a pair of exact solutions $(w_1^-, w_2^-)$ to the equation $(P-E)w=0$, which is a basis of microlocal solutions near $(0,0)$ (this means that $w_1^-$ and $w_2^-$ are bounded as $h\to +0$ both on the left and right of $x=0$). We express $(w_1^-, w_2^-)$ in linear combination \eqref{coeff} of the basis $(w_1^\flat, w_2^\flat, w_1^\sharp, w_2^\sharp)$. Then we can deduce the transfer matrix from this $4\times 2$ matrix since we know in fact the microlocal behavior of $(w_1^\flat, w_2^\flat, w_1^\sharp, w_2^\sharp)$ on the incoming and outgoing trajectories.
Thus we obtain two term asymptotics of the transfer matrix. In particular, the subprincipal term, which will give the resonance width, is reduced to an integral \eqref{IEh} of the product of two WKB solutions associated with $P_1$ and $P_2$. This oscillatory integral has its degenerate critical point of order 
$n$ at the crossing point, i.e. the derivative of the phase vanishes  at order $n$. 
Contrary to \cite{AFH1}, this critical point is close to the turning points.
The two turning points are of order $E$, and their difference is of order $E^n$. Hence if $E$ is sufficiently small ($\re E=o(h^{4/(2n+1)})$, see Remark~\ref{rem1} and Section~\ref{sec:small}) and $n\ge 2$, roughly speaking, the phase  behaves like 
$x^{n+\frac 12}$ and the amplitude  like $x^{-\frac 12}$ near the crossing-turning point $x=0$. Then after the change of variable $x=z^2$,
the phase becomes like $z^{2n+1}$ and the singularity of the amplitude dissappears. The degenerate stationary phase in the variable $z$ gives the asymptotics of order $h^{1/(2n+1)}$ of the integral. For the maximal energy range $\re E=\ord(h^{2/(2n+1)})$, the asymptotic behavior of this integral is given by a generalization of the Airy function. This describes the transition of the order of the integral from $h^{1/(n+1)}$ to $h^{1/(2n+1)}$ due to the confluence of crossing points.

This article is organized as follows. The main results (Theorem \ref{MAINTH} and Theorem \ref{prop:crossing-turning}) are stated in Section \ref{Section 2}. Section \ref{app:Connection} is devoted to the proof of Theorem \ref{prop:crossing-turning}. In section \ref{sec:small}, we restrict the real part of the resonances in a smaller interval to obtain a constant coefficient of the leading term and a better error estimate for the resonance asymptotics.
In Appendix \ref{airyapp}, we note some fundamental properties of Airy functions and the generalized Airy function. Finally in Appendix \ref{Appendix B}, we
review various solutions to the single Schr\"odinger equation near a simple turning point that we use in the construction of the exact solutions to our system.


\section{Assumptions and results} \label{Section 2}
We define the resonances of the operator $P$  under the following condition.
\begin{As}\label{H2} 
The functions $V_1,V_2,r_0$ and $r_1$ are real-valued analytic functions on $\mathbb R$  and extend to bounded holomorphic functions in the complex domain
\ben
\Sigma:=\{x\in\C;\,\left|\im x\right|< \delta_0 \langle \re x\rangle\}
\een
for some constant $\delta_0>0$, where $\langle t \rangle:= (1+t^2)^{1/2}$. Moreover, for each $j=1,2$, $V_j$ admits non-zero limits $V_j^\pm\neq0$ as $\re x\to \pm\infty$ in $\Sigma$. 
\end{As}

For $\delta_1,\delta_2>0$ possibly $h$-dependent, one introduces the complex box 
\be\label{DEFR}
\cR=\cR(\delta_1,\delta_2) :=[-\delta_1,\delta_1] +i[-\delta_2,\delta_2].
\ee
Under Assumption~\ref{H2}, the operator $P$ is self-adjoint in $L^2(\mathbb R;\mathbb C^2)$ with domain $H^2(\mathbb R;\mathbb C^2)$. One can define the resonances of $P$ in $\cR$, for small $\dl_1$ and $\dl_2$,  
as the values $E\in\C$ such that the equation $Pw=Ew$ has a non-trivial outgoing solution $w$, that is, a non-identically vanishing solution such that for some small $\theta>0$, the function $x\mapsto w(x e^{i\theta})$ is in $L^2(\R;\C^2)$ (see e.g., \cite{AgCo}). Equivalently, the resonances can be defined as the eigenvalues of the operator $P$ acting on $L^2(\mathbb R_\theta;\mathbb C^2)$, where $\mathbb R_\theta$ is a complex distortion of $\mathbb R$ that coincides with $e^{i\theta} \mathbb R$ for $x$ large enough (see e.g., \cite{HeMa}).
In particular, the imaginary part of each resonance is negative.  We denote $\ope{Res}(P)$ the set of resonances of $P$.

We study resonances near a crossing level $E_0=0$ of a simple model (see Figure \ref{Fig}).
\begin{As}\label{H1}
There exists $b_0<0$ such that 
\be\label{eq:modelsV}
\frac{V_1(x)}{x(x-b_0)}>0,\quad
\frac{V_2(x)}{x}>0\qtext{for}x\in\R,
\qtext{and}
V_2(x)<V_1(x)
\qtext{for}x<0.
\ee
We also assume that the vanishing order at $x=0$ of $V_1-V_2$ is finite.
\end{As}
\begin{figure}
\centering\includegraphics[bb=0 0 470 590, width=8cm]{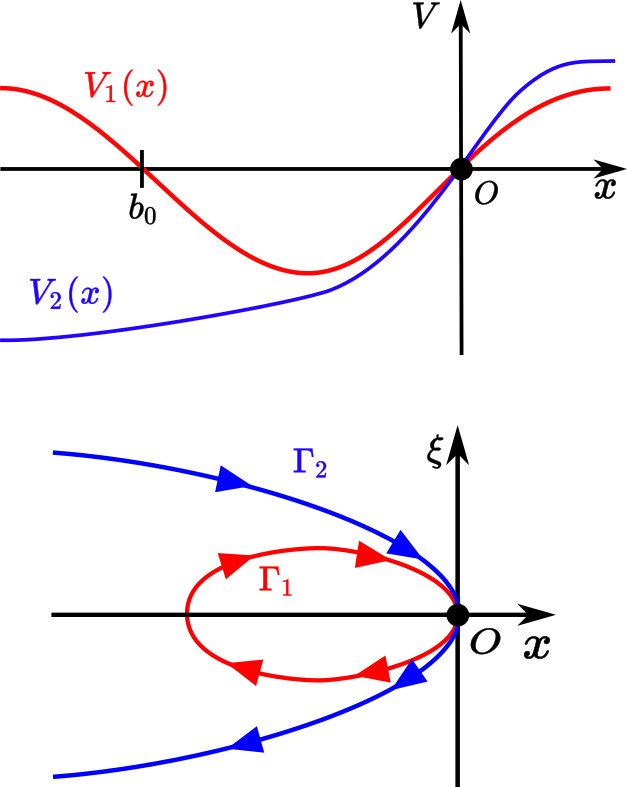}
\caption{The potentials and the classical trajectories}\label{Fig}
\end{figure}

 In the following, we denote by $n\in \mathbb N\setminus \{0\}$ the vanishing order of $V_1-V_2$ at $x=0$:
\be\label{contact-o}
V_1^{(k)}(0)-V_2^{(k)}(0)=0\quad(0\le k\le n-1),\qquad
V_1^{(n)}(0)-V_2^{(n)}(0)\neq0.
\ee

 Let $H_{p_j}:=2\xi\p_x-V_j'(x)\p_\xi$ be the Hamiltonian vector-field associated with $p_j$, $j=1,2$. For an energy $E\in \mathbb R$, denote by $\Gamma_j(E)$ the characteristic set associated with $P_j$ given by 
\ben
\Gamma_j(E):=p_j^{-1}(E)=\{(x,\xi)\in T^*\R;\,p_j(x,\xi)=E\} \quad (j=1,2),
\een
and set 
\ben
\Gamma(E):= \Gamma_1(E) \cup \Gamma_2(E).
\een
At $E=0$, we simply write $\Gamma_j:=\Gamma_j(0)$, $j=1,2$.  Assumption \ref{H1} means that $V_1$ admits a simple well at the energy-level 0 and hence  $\Gamma_1$ is a simple closed curve symmetric with respect to the $x$-axis, while 0 is a non-trapping energy for $p_2$.
For $E$ near $0$,  we define the action integral  
\ben\label{AAC}
\cA(E):=  \int_{\Gamma_1(E)} \xi dx.
\een
The function $\mathcal{A}(E)$ extends to an analytic function in a complex neighborhood of $E=0$.  Under Assumption \ref{H1}, it is well known that the spectrum of the scalar Schr\"odinger operator $P_1$ in a small neighborhood of $0$ consists of eigenvalues approximated by the roots of the Bohr-Sommerfeld quantization rule (see e.g., \cite{ILR, Ya}) 
\begin{equation}\label{BSR}
\cos\left( \frac{\mathcal{A}(E)}{2h}\right) = 0.
\end{equation}
For $\delta_1,h>0$ small enough, we define an $h$-dependent discrete set $\mathfrak{B}_h\subset [-\delta_1, \delta_1]$ by 
\begin{equation}\label{DEFFF}
\mathfrak{B}_h=\mathfrak{B}_{h}(\delta_1) := \left\{E\in  [-\delta_1, \delta_1];\, \cos\left( \frac{\mathcal{A}(E)}{2h}\right) = 0 \right\}.
\end{equation}
The distance of two successive points in $\mathfrak{B}_h$ is of order $h$, more precisely, 
\ben
\ope{dist}(\mathfrak{B}_h\setminus\{E\},E)=\frac {2\pi h}{\cA'(E)}+\ord(h^2)\quad \text{for any }E\in \mathfrak{B}_h.
\een
In the phase space, Assumption \ref{H1} implies that $\Gamma_1$ and $\Gamma_2$ cross tangentially at $(0,0)\in T^*\R$ with contact order $2n$, more precisely, we have for any $n\geq 2$
$$
H_{p_1}^kp_2(0,0)=0\quad(0\le k\le 2n-1),
$$
$$
H_{p_1}^{2n}p_2(0,0) = (-1)^n  \frac{(2n)!}{n!} (V_1'(0))^n (V_2^{(n)}(0) -V_1^{(n)}(0)) \neq 0.
$$

The following result, which is our main result, describes the asymptotic distribution of resonances in 
$\mathcal R(Lh^{\frac2{2n+1}}, Lh)$ for an arbitrarily fixed $L>0$. Their imaginary part is of polynomial order in $h$ with power $\frac{2n+3}{2n+1}=1+\frac{2}{2n+1}$. The coefficient of this leading term is expressed in terms of a generalization $A_n$ of the Airy function (see Appendix \ref{airyapp})
\be\label{defAn}
A_n(y)
:=\frac1{2\pi}\int_{\R}\exp\left(i\int_0^\eta(y+\tau^2)^n d\tau\right)d\eta.
\ee
\begin{theorem}\label{MAINTH}
Suppose Assumptions \ref{H2} and \ref{H1}, and let $\mathfrak{B}_h$ and $\mathcal R$ be given by \eqref{DEFFF} and \eqref{DEFR} with $\delta_1=Lh^{\frac2{2n+1}}$, $\delta_2=Lh$ for an arbitrarily fixed $L>0$. Then, 
for $h>0$ small there exists a bijective map 
\ben
z_h : \mathfrak{B}_h \to {\rm Res} (P)\cap {\mathcal R}
\een 
such that for any $E\in \mathfrak{B}_h$ one has 
\begin{equation}\label{Realpart}
\vert z_h(E) - E \vert =  \mathcal{O}(h^{\frac{2n+3}{2n+1}}),
\end{equation}
\begin{equation}\label{Impart}
{\rm Im}\, z_h(E) = -
\frac{\kappa_n(\lambda)^2}{\cA'(E)}h^{\frac{2n+3}{2n+1}}+\ord(h^{\frac{2n+4-\epsilon}{2n+1}}),
\end{equation}
uniformly as $h\to0^+$, with $\epsilon=0$ when $n=1$, $\frac12$ when $n=2$, and $\frac3{4n+3}$ when $n\ge 3$. The coefficient $\kappa_n$ is a function of $\lambda:=E/h^{\frac 2{2n+1}}$ given by 
\be
\label{kpn}
\kappa_n(\lambda)=\frac{2\pi r_0(0)}{(V_1'(0) V_2'(0))^{1/4}}q_{n}^{-\frac 1{2n+1}}{A}_{n}\left (-\frac{q_{n}^{\frac{2}{2n+1}}}{(V_1'(0) V_2'(0))^{1/2}}\lambda\right ), 
\ee
where $A_n$ is defined by \eqref{defAn},
and $q_n$  is the constant given by 
\be\label{qn}
q_n=\frac{V_1^{(n)}(0)-V_2^{(n)}(0)}{n!(V_1'(0) V_2'(0))^{1/4}}. 
\ee
\end{theorem}
\noindent
Note that $V_1'(0)V_2'(0)>0$ from Assumption~\ref{H1}.

\begin{remark}
Strictly speaking,  in Theorem \ref{MAINTH},  the semiclassical limit $h\to 0^+$ has to be taken in a subset $\frak{I}\subset]0,1]$ with $0\in\overline{\frak{I}}$ to ensure the bijectivity of the map $z_h$. 
When a Bohr-Sommerfeld point $E\in \mathfrak{B}_h$ is very close to one of the endpoints $\pm \dl_1=\pm Lh^{2/(2n+1)}$ of the interval $[-\dl_1,\dl_1]$, the bijectivity may break down. This is a usual problem (see e.g., \cite{HeSj1,HeSj2}) in the semiclassical distribution of resonances/eigenvalues in some fixed domain.
We simply write $h\to 0^+$ in order not to complicate the statement (as in \cite{AFH1}).

\end{remark}

\begin{remark}
The generalized Airy function $A_n$ appearing in the coefficient of the leading term explains the transition from positive energies with intersection  of $\Gamma_1$ and $\Gamma_2$ to negative ones without intersection. In fact, when $L$ is large, the variable  $y=-\lambda q_{n}^{2/(2n+1)}/\sqrt{V_1'(0)V_2'(0)}$  in this function (see \eqref{kpn}) is a positive large or negative large number near the negative or positive extremity of the interval $[-\delta_1, \delta_1]$ respectively. As we see in Appendix {\ref{airyapp}}, $A_n(y)$ decays exponentially as $y\to +\infty$, whereas it is oscillating as $y\to -\infty$. 
\end{remark}

\begin{remark}
Our method applies also to the case  $V_1'(0)V_2'(0)<0$ with $n=1$, which was studied in \cite{FMW1} without microlocal argument. 
To include this case, it is enough to put absolute value for $V_1'(0)V_2'(0)$ in \eqref{kpn} and \eqref{qn}. This coincides with the result in \cite[$(2.5)$]{FMW1}  by modifying the misprint $\mu_1(\lambda_k(h))^2+\mu_2(\lambda_k(h))^2$ with $ |V_1'(0)V_2'(0)|^{-\frac13}\left[\mu_1(\lambda_k(h))+\mu_2(\lambda_k(h))\right]^2$.
\end{remark}


As in  \cite{AFH1}, the proof of Theorem~\ref{MAINTH} is essentially reduced to the following microlocal connection formula at the crossing point $(0,0)$.
Here we only assume the following local condition:
\begin{figure}
\centering
\includegraphics[bb=0 0 190 223, width=5.5cm]{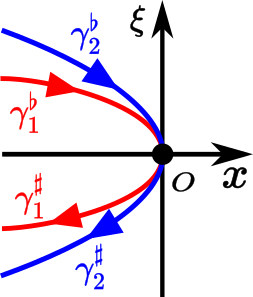}
\caption{Incoming and outgoing trajectories near $(0,0)$}
\label{Fig:2}
\end{figure}
\begin{condition}\label{H3}
The functions $V_1$, $V_2$, $r_0$ and $r_1$ are real-valued analytic functions in a small interval around the origin $x=0$ and they extend to holomorphic functions in a small complex neighborhood of the origin. Moreover, $V_1'(0)>0$, $V_2'(0)>0$ and there exists $n\in\N\setminus\{0\}$ such that \eqref{contact-o} holds.
\end{condition}


In the following we use the same microlocal and semiclassical notations and terminologies as in \cite{AFH1}. In particular, we recall that if $f=f(x,h)\in L^2(\mathbb R)$ with $\Vert f\Vert_{L^2} \leq 1$, the notation $f\equiv 0$ near a point/subset of $T^*\mathbb R$ means that $f$ is microlocally $0$ near that point/subset.

For each $j=1,2$, let $\gamma_j^\flat,\gamma_j^\sharp$ be respectively the incoming and outgoing trajectories on $\Gamma_j$ in a neighborhood of $(x,\xi)=(0,0)$ (see Figure~\ref{Fig:2}). Remark that, for $j=1,2$, one has
$$
\gamma_j^\flat\subset \{(x,\xi);x<0, \xi>0\},\quad
\gamma_j^\sharp\subset \{(x,\xi);x<0, \xi<0\}.
$$
Following \cite[Proposition 5.4]{FMW3}, with each curve  $\gamma_j^\dir$, 
we associate a microlocal WKB solution $f_j^\dir$, where $\bullet$ stands for $\flat$ or $\sharp$:
\begin{equation}
(P-E)f_j^\dir\equiv 0\quad
\text{near }\gamma_j^\dir,
\end{equation}
which is
of the form
\begin{equation}
f_j^\dir= e^{i\phi_j^\dir (x)/h}\begin{pmatrix}
\sigma_{j,1}^\dir(x,h)\\
\sigma_{j,2}^\dir(x,h)
\end{pmatrix},\quad
\sigma_{j,k}^\dir(x,h)\sim\sum_{l\ge0}h^l\sigma_{j,k,l}^\dir(x)\quad k=1,2.
\end{equation}
For $x<0$ small, the phase function $\phi_j^\dir$ is given by
$$\phi_j^\flat(x):=\int_{a_j(E)}^x\sqrt{E-V_j(y)}\,dy,\quad \phi_j^\sharp(x):=-\int_{ a_j(E)}^x\sqrt{E-V_j(y)}\,dy$$ 
where $a_j(E)$ is the unique solution of $V_j(x)=E$ near $x=0$.  The symbols $\sigma_{j,k}^\dir$ satisfies
\be
\sigma_{j,j,0}^\dir(x)=(E-V_j(x))^{-1/4}, \quad
\sigma_{j,\hat j,0}^\dir(x)=0 \quad  (\hat j=3-j). 
\ee
The space of microlocal solutions on a classical trajectory is one-dimensional except at the crossing point, and hence $f_j^\dir$ is a basis on $\gamma_j^\dir$.




\begin{theorem}\label{prop:crossing-turning}
Suppose Assumption~\ref{H3}, and let $E$ belong to $\cR=\cR(Lh^{\frac{2}{2n+1}},Lh)$ for an arbitrary $h$-independent $L>0$. There exists a $2\times 2$ matrix $T=T(E,h)$ such that
if $u\in L^2(\R;\C^2)$ with $\|u\|_{L^2}\le1$ satisfies
\begin{equation}
(P-E)u\equiv 0
\end{equation}
in a neighborhood of $(0,0)$, and if 
\begin{equation}
\label{gammaflat}
u\equiv \alpha_j^\dir f_j^\dir\quad \text{on }\gamma_j^\dir\quad(j=1,2,\ \dir=\inc,\out),
\end{equation}
for some constants $\alpha_j^\dir = \alpha_j^\dir (u,E,h)\in \mathbb C$, then we have
\be\label{transfer}
\biggl(\begin{matrix}
\alpha_{1}^\sharp\\
\alpha_{2}^\sharp
\end{matrix}\biggr)
=
T(E,h)
\biggl(
\begin{matrix}
\alpha_{1}^\flat\\
\alpha_{2}^\flat
\end{matrix}\biggr).
\ee
Moreover, $T$ satisfies the asymptotic behavior
\be\label{eq:Thm2}
iT(E,h)= {\rm Id}_2 - i\kappa_n (\lambda) h^{\frac 1{2n+1}}\begin{pmatrix} 
0& 1\\
1&0
\end{pmatrix}+\ord(h^{\frac {2-\epsilon}{2n+1}
}), 
\ee
uniformly with respect to $h>0$ small enough, where ${\rm Id}_2$ is the $2\times 2$ identity matrix, $\lambda=E/h^{\frac{2}{2n+1}}$, $\kappa_n$ is given by \eqref{kpn}, and $0\le\epsilon<1$ is given in Theorem~\ref{MAINTH}.
\end{theorem}

\begin{remark}[Improved asymptotics in a smaller region]\label{rem1} If we reduce the real part of the domain ${\mathcal R}$ with $\dl_1=o(h^{\frac4{2n+1}})$ and if $n\ge2$, then we have the following better asymptotic formulae than \eqref{eq:Thm2} and \eqref{Impart}. In this case, we have instead of formula \eqref{eq:Thm2} 
\begin{equation}\label{FIM1}
iT(E,h)={\rm Id}_2- i\kappa_n(0) h^{\frac 1{2n+1}}\begin{pmatrix}
0& 1\\
1&0
\end{pmatrix}+\ord\left(h^{\frac2{2n+1}}+h^{\frac13}\right),
\end{equation}
uniformly with respect to $h>0$ small enough. Moreover, formula \eqref{Impart} is improved to the following one with a constant leading coefficient and a smaller error 
\begin{equation}\label{FIM2}
{\rm Im}\, z_h(E) = -
\frac{\kappa_n(0)^2}{\cA'(0)}h^{\frac{2n+3}{2n+1}}+\ord(h^{\frac{2n+4}{2n+1}}),
\end{equation}
uniformly as $h\to0^+$. The constant $\kappa_n(0)$ is given by
\begin{equation}\label{eq:kappa-m0}
\kappa_n(0)= 2r_0(0) \left(\frac{(2n+1)!}{\vert H_{p_1}^{2n} p_2(0,0)\vert}\right)^{\frac{1}{2n+1}}        \bm{\Gamma}\left(\frac{2n+2}{2n+1}\right)\cos\left(\frac\pi{2(2n+1)}\right),
\end{equation}
where $\bm{\Gamma}$ is the usual Gamma function defined by \eqref{Gf}.  
\end{remark}

\begin{remark}
We finish this section by comparing our result with the one in \cite{AFH1} where the energy is above the crossing level. The coefficient $\kappa_n(0)$ in \eqref{FIM1} and \eqref{FIM2} coincides with the one appearing in the formula (3.8) of \cite{AFH1} on the asymptotics of the microlocal transfer matrix (called $\omega_\rho$ there), more precisely $\kappa_n(0)= \omega_{(0,0)}$. Notice also that the error $h^{\frac{2n+4}{2n+1}}$ in \eqref{FIM2} coincides with that in \cite{AFH1} if $n\geq 2$. Finally, it is worth mentioning that the multiplication by the complex number $i$ in the asymptotics of the transfer matrix \eqref{eq:Thm2} is related to the Maslov index since the crossing occurs at a turning point which was excluded in our previous work.
\end{remark}

\section{Proof of the microlocal connection formulae}\label{app:Connection}




%
We first review  the construction of exact solutions due to \cite{FMW1}. 
%
If $(u, u^\dagger)$ is a pair of independent solutions to the scalar homogeneous equation $(-h^2\frac {d^2}{dx^2}+V(x)-E)u=0$, the
solutions to the inhomogeneous equation $(-h^2\frac {d^2}{dx^2}+V(x)-E)v=f$ are written in the form
\be\label{eq:Def-K-Gen}
K(u,u^\dagger;s,t)[f](x):=\frac{1}{h^2\W(u,u^\dagger)}\left(u(x)\int_{s}^x u^\dagger(y)f(y)dy-u^\dagger(x)\int_{t}^x u(y)f(y)dy\right)
\ee
for any $f\in C_0^\infty(\R)$ with arbitrary constants $s,t\in [-\infty,+\infty]$. Here $\W(u,u^\dagger):={\rm det}(u,u^\dagger)$ stands for the Wronskian of $u$ and $u^\dagger$. 

In the following we use the notation $U_1:=U$ and $U_2 := U^*$. Let $(u_j,u^\dagger_j)$ be a pair of independent solutions to $(P_j-E)u=0$ and set $K_j:=K(u_j,u^\dagger_j; s,t)$ and $\til{K}_j:=-hK_jU_j$.  We suppose here that the interaction functions $r_0(x)$ and $r_1(x)$ have a compact support near $x=0$. Then the operator $\widetilde K_j
$ is applied to smooth functions without compact support.

Remark that the function $\til K_j[f](x)$ and its derivative take at $x=0$ the following values:
\begin{align}
\til K_j[f](0)=C_j(u_j,u^\dagger_j;s)[f]u_j(0)+C_j(u^\dagger_j, u_j;t)[f]u^\dagger_j(0), 
\label{Ku}\\
(\til K_j[f])'(0)=C_j(u_j,u_j^\dagger;s)[f]u_j'(0)+C_j(u^\dagger_j, u_j;t)[f](u^\dagger_j)'(0),
\label{Ku'}
\end{align}
where
$$
C_j(u,u^\dagger;s)[f]:=\frac{1}{h\W(u,u^\dagger)}\int^{s}_0 u^\dagger(y)U_jf(y)dy.
$$
Notice that $C_j(u,v;0)=0$ is a trivial functional.

If the infinite sums
\be\label{eq:Inf-Sum}
J_1=\sum_{k\ge0}(\til{K}_1\til{K}_2)^k,\quad J_2=\sum_{k\ge0}(\til{K}_2\til{K}_1)^k
\ee 
converge in some suitable space, then 
\be\label{Def-Sol-Gen}
w_1=w_1(K_1,K_2,u_1):=
\begin{pmatrix}J_1u_1\\\til{K}_2J_1u_1\end{pmatrix},\quad
w_2=w_2(K_1,K_2,u_2):=
\begin{pmatrix}\til{K}_1J_2 u_2  \\J_2u_2\end{pmatrix},
\ee
 are exact solutions to the system
 \be
 \label{system}
 (P-E)u=0.
 \ee

For $j=1,2$, let $(u^\out_j ,u^\inc_j )$ be the outgoing and incoming solutions on the left of the origin, and $(u^-_j,u^+_j)$ the exponentially decaying and growing solutions on the right to $(P_j-E)u=0$ given in \eqref{eq:out-inc-Airy-sol} and \eqref{eq:tj-scalar-pos}. 
Using these solutions, we define exact solutions to the system
\begin{align*}
&w_j^-:=w_{j}(K_1^{\rm p},K_2^{\rm p},u_j^-)
&&\text{with }
K_j^{\rm p}:=K(u_j^-,u_j^+,0,\infty),\quad &&j=1,2, \\
&w_j^\dir :=w_j(K_1^{\rm a},K_2^{\rm a},u_j^\dir)
&&\text{with }
K_j^{\rm a}:=K(u_j^\out,u_j^\inc,-\infty,-\infty),\quad &&j=1,2,\,\,\,\dir=\out,\inc.
\end{align*}
We denote 
$J_1^{\rm p}:=\sum (\til{K}^{\rm p}_1\til{K}^{\rm p}_2)^k$, $J_1^{\rm a}:=\sum (\til{K}^{\rm a}_1\til{K}^{\rm a}_2)^k$ and so on.  For $\{j,\hat j\}=\{1,2\}$, $(\dir,\dir')\in\{\out,\inc\}^2$ and $\{\dir,\hat\dir\}=\{\out,\inc\}$, we set
\begin{align}
\label{al}
&\alpha_j :=C_{\hat j}(u_{\hat j}^+,u_{\hat j}^-;+\infty)[J_j^{\rm p}u_j^-],\quad
\alpha_{j,\dir'}^{\dir} :=C_{\hat j}(u_{\hat j}^{\dir},u_{\hat j}^{\hat\dir};-\infty)[J_j^{\rm a}u_j^{\dir'}],\\
\label{be}
&\beta_j :=C_j(u_{ j}^+,u_{j}^-;+\infty)[\til K_{\hat j}^{\rm p}J_j^{\rm p}u_j^-],\quad
\beta_{j,\dir'}^{\dir} :=C_{j}(u_{j}^{\dir},u_{j}^{\hat\dir};-\infty)[\til K_{\hat j}^{\rm a}J_j^{\rm a}u_j^{\dir'}]. 
\end{align} 
These quantities and the estimates given in the Lemma below are important in our study of the microlocal connection formula. In the following, $C(I_L)$ and $C(I_R,u_{\hat j}^-)$ stand respectively for the Banach spaces of continuous functions on $I_L:=(-\infty,0]$ and $I_R:=[0,+\infty)$ with the norms given by $\|f\|_{C(I_L)} :=\sup_{I_L}\left|f\right|$ and $\|f\|_{C(I_R,u_{\hat j}^-)} :=\sup_{I_R}\left|f/u_{\hat j}^-\right|$, and  $\|\cdot\|_{\mc{B}(B)}$ stands for the norm of bounded operators on a Banach space $B$. 
\begin{lemma}
\label{estimate}
We have the following estimates uniformly as $h\to0^+$ for $E\in\cR$:
\be
\label{k1k2}
\left\|\til{K}_1^{\rm a}\til{K}_2^{\rm a}\right\|_{\mc{B}(C(I_L))}=\ord(h^{\frac1{2n+1}}),\quad
\left\|\til{K}_2^{\rm a}\til{K}_1^{\rm a}\right\|_{\mc{B}(C(I_L))}=\ord(h^{\frac1{2n+1}}),
\ee
\be
\label{kj}
\left\|\til{K}_j^{\rm p}\right\|_{\mc{B}(C(I_R,u_{\hat{j}}^-))}=
\ord(h^{\frac13}),
\ee
Moreover, for  $j=1,2$, $(\dir,\dir')\in\{\out,\inc\}^2$, one has 
\be
\label{aa}
\alpha_{j,\dir'}^{\dir}=\ord(h^{\frac1{2n+1}}),\quad
\beta_{j,\dir'}^{\dir}=\ord(h^{\frac2{2n+1}}),
\ee
\be
\label{a-}
\alpha_j =\ord(h^{\frac13}),\quad
 \beta_j= \ord(h^{\frac 23}).
\ee
\end{lemma}

\begin{proof}
This kind of argument has already been done in \cite{FMW1} and \cite{AFH1}. Here we just sketch the essence of the proof. The integral kernels of the operators $\til{K}_1^{\rm a}\til{K}_2^{\rm a}$ and $\til{K}_2^{\rm a}\til{K}_1^{\rm a}$ are oscillatory integrals similar to that in Proposition~\ref{prop:key-int}. This fact with the estimate in Proposition~\ref{prop:key-int} shows \eqref{k1k2}. 
The estimates \eqref{aa} follows from the same oscillatory integrals. Note that we have two integrals in $\beta_{j,\dir'}^{\dir}$. This is the reason why the order is the square of that for $\alpha_{j,\dir'}^{\dir}$.

The estimate for the operator $\til{K}_j^{\rm p}$ is well studied  in \cite{As}. We use the Laplace method and the fact that the function $u_j^{\pm}$ is of $\ord(h^{-1/6})$ in a $\ord(h^{2/3})$-neighborhood of the origin. The latter gives a larger contribution of order $h^{1/3}$ to the norm. This gives the estimates \eqref{kj} and \eqref{a-}.
\end{proof}

The family $(w_1^\inc,w_2^\inc,w_1^\out,w_2^\out)$ is a basis of local solutions in a neighborhood of $x=0$. 
On the other hand, the pair $(w_1^-, w_2^-)$ can be regarded as a basis of microlocal solutions in a neighborhood of $(x,\xi)=(0,0)$, since 
they are bounded as $h\to 0^+$ on both sides of the origin. Remark that this is not true for $w_1^\inc,w_2^\inc,w_1^\out,w_2^\out$. They are exponentially large on the right.

We write $(w_1^-,w_2^-)$ in the linear combination of $(w_1^\inc ,w_2^\inc ,w_1^\out  ,w_2^\out )$:
\be
\label{coeff}
(w_1^-,w_2^-)=(w_1^\inc ,w_2^\inc ,w_1^\out  ,w_2^\out )\begin{pmatrix}
A^\inc(E,h) \\ A^\out(E,h)
\end{pmatrix},
\ee
where $A^\inc$, $A^\out$ are  $2\times 2$ matrices.
The microlocal asymptotic behavior of each solution on each $\gamma_j^\dir$, away from the crossing point, is given for $j=1,2,$ $\dir=\inc,\out$ by
\be\label{eq:AlmostIB}
w_j^\dir\equiv 
\left\{
\begin{array}{l}
f_{\gamma_j^\dir}\qtext{near }\gamma_j^\dir, \\
\ord(h)\qtext{near }(\gamma_1^\inc\cup\gamma_1^\out\cup\gamma_2^\inc\cup\gamma_2^\out)\setminus\gamma_j^\dir.
\end{array}
\right.
\ee
This implies that the transfer matrix $T$ is given by
\be 
\label{T}
T(E,h)=A^\out(E,h)A^\inc(E,h)^{-1}+\ord(h).
\ee
In fact, we have
$$
(w_1^-,w_2^-)A^\inc(E,h)^{-1}=(w_1^\inc ,w_2^\inc ,w_1^\out  ,w_2^\out )\begin{pmatrix}
{\rm Id}_2  \\ A^\out(E,h)A^\inc(E,h)^{-1}
\end{pmatrix},
$$ 
which means that there exist two microlocal solutions
\be\label{inc-sol}
w_1^\inc+\tau_{11}w_1^\out+\tau_{21}w_2^\out,\quad w_2^\inc+\tau_{12}w_1^\out+\tau_{22}w_2^\out,
\ee
where $\tau_{jk}$ is the $(j,k)$-entry of the matrix $A^\out(E,h)A^\inc(E,h)^{-1}$. These microlocal solutions satisfy 
\eqref{gammaflat} 
modulo $\ord (h)$ with 
$(\alpha_1^\flat, \alpha_2^\flat, \alpha_1^\sharp, \alpha_2^\sharp)=
(1,0,\tau_{11}, \tau_{21})$ and $(0,1,\tau_{12}, \tau_{22})$ respectively.


Let us now compute the semiclassical asymptotics of the matrices $A^\out(E,h)$ and $A^\inc(E,h)$ defined by \eqref{coeff}.
To do this, we make use of the identity at the origin $x=0$:
\be
\begin{pmatrix}
\label{origin}
w_1^-&w_2^-\\
(w_1^-)'&(w_2^-)'
\end{pmatrix}|_{x=0}=
\begin{pmatrix}
w_1^\inc &w_2^\inc&w_1^\out&w_2^\out\\
(w_1^\inc)' &(w_2^\inc)'&(w_1^\out)'&(w_2^\out)'
\end{pmatrix}|_{x=0}\begin{pmatrix}
A^\inc(E,h) \\ A^\out(E,h)
\end{pmatrix}.
\ee

We have
$$
w_1^-=\begin{pmatrix}J_1^{\rm p}u_1^-\\\til{K}_2^{\rm p}J_1^{\rm p}u_1^-\end{pmatrix},\quad
w_2^-=\begin{pmatrix}\til{K}_1^{\rm p}J_2^{\rm p}{u}_2^-\\J_2^{\rm p}u_2^-\end{pmatrix},
\quad
w_1^\bullet=\begin{pmatrix}J_1^{\rm a}u_1^\bullet\\\til{K}_2^{\rm a}J_1^{\rm a}u_1^\bullet\end{pmatrix},\quad
w_2^\bullet=\begin{pmatrix}\til{K}_1^{\rm a}J_2^{\rm a}{u}_2^\bullet\\J_2^{\rm a}u_2^\bullet\end{pmatrix}.
$$
The values of these vector-valued functions and their derivatives at the origin, say of $w_1^-$, can be written using \eqref{Ku} and \eqref{Ku'} in the form
$$
w_1^-(0)=\begin{pmatrix}u_1^-(0)+\til{K}_1^{\rm p}(\til{K}_2^{\rm p}J_1^{\rm p}u_1^-)(0)\\
\til{K}_2^{\rm p}(J_1^{\rm p}u_1^-)(0)\end{pmatrix}
=\begin{pmatrix}u_1^-(0)+\beta_1u_1^+(0)\\
\alpha_1u_2^+(0)\end{pmatrix},
$$
$$
(w_1^-)'(0)
=\begin{pmatrix}(u_1^-)'(0)+\beta_1(u_1^+)'(0)\\
\alpha_1(u_2^+)'(0)\end{pmatrix},
$$
Similarly, we can express $w_2^-$, $w_1^\inc$, $w_2^\inc$, $w_1^\out$ and $w_2^\out$ in terms of the constant vectors
$$
{\bf u}_1^\pm=\begin{pmatrix}u_1^\pm(0)\\ 0\\ (u_1^\pm)'(0)\\ 0\end{pmatrix},\quad
{\bf u}_2^\pm=\begin{pmatrix}0 \\u_2^\pm(0)\\ 0\\ (u_2^\pm)'(0)\end{pmatrix},\quad
{\bf u}_1^\bullet=\begin{pmatrix}u_1^\bullet(0)\\ 0\\ (u_1^\bullet)'(0)\\ 0\end{pmatrix},\quad
{\bf u}_2^\bullet=\begin{pmatrix}0 \\u_2^\bullet(0)\\ 0\\ (u_2^\bullet)'(0)\end{pmatrix},
$$
and obtain
\begin{equation}
\label{wpm}
\begin{pmatrix}
w_1^-&w_2^-\\
(w_1^-)'&(w_2^-)'
\end{pmatrix}|_{x=0}
=\begin{pmatrix}
{\bf u}_1^-&{\bf u}_2^-&{\bf u}_1^+&{\bf u}_2^+
\end{pmatrix}
\begin{pmatrix}
1	&0	\\
0	&1	\\
\beta_1	&\alpha_2	\\
\alpha_1	&\beta_2	
\end{pmatrix},
\end{equation}
\begin{align}
&\begin{pmatrix}
w_1^\inc &w_2^\inc&w_1^\out&w_2^\out\\
(w_1^\inc)' &(w_2^\inc)'&(w_1^\out)'&(w_2^\out)'
\end{pmatrix}|_{x=0} \\
&\label{wfs}
=\begin{pmatrix}
{\bf u}_1^\inc&{\bf u}_2^\inc&{\bf u}_1^\out&{\bf u}_2^\out
\end{pmatrix}\begin{pmatrix}
1+\beta_{1,\inc}^\inc 	&\alpha_{2,\inc}^\inc 	&\beta_{1,\out}^\inc 	&\alpha_{2,\out}^\inc \\
\alpha_{1,\inc}^\inc 	&1+\beta_{2,\inc}^\inc 	&\alpha_{1,\out}^\inc 	&\beta_{2,\out}^\inc \\
\beta_{1,\inc}^\out 	&\alpha_{2,\inc}^\out 	&1+\beta_{1,\out}^\out 	&\alpha_{2,\out}^\out \\
\alpha_{1,\inc}^\out 	&\beta_{2,\inc}^\out 	&\alpha_{1,\out}^\out 	&1+\beta_{2,\out}^\out 
\end{pmatrix}.
\end{align}
On the other hand, we have the relation \eqref{eq:tj-scalar-pos} between the solutions to the scalar Schr\"odinger equations which yields
\begin{equation}
\label{pmflat}
\begin{pmatrix}
{\bf u}_1^-&{\bf u}_2^-&{\bf u}_1^+&{\bf u}_2^+
\end{pmatrix}=
\begin{pmatrix}
{\bf u}_1^\inc&{\bf u}_2^\inc&{\bf u}_1^\out&{\bf u}_2^\out
\end{pmatrix}\frac{1}{2}
\begin{pmatrix}
\bar{t}_1&0&t_1&0\\
0&\bar{t}_2&0&t_2\\
t_1&0&\bar{t}_1&0\\
0&t_2&0&\bar{t}_2
\end{pmatrix},
\end{equation}
with $t_j=e^{-\frac\pi 4i}+\ord (h)$ for $j=1,2$.

Summing up, we obtain from \eqref{aa}, \eqref{a-}, \eqref{origin}, \eqref{wpm}, \eqref{wfs} and \eqref{pmflat}
\begin{equation}
\label{matrixproduct}
\begin{pmatrix}
A^\inc \\ A^\out
\end{pmatrix}=\begin{pmatrix}
1 	&\alpha_{2,\inc}^\inc 	&0 	&\alpha_{2,\out}^\inc \\
\alpha_{1,\inc}^\inc 	&1 	&\alpha_{1,\out}^\inc 	&0 \\
0 	&\alpha_{2,\inc}^\out 	&1 	&\alpha_{2,\out}^\out \\
\alpha_{1,\inc}^\out 	&0 	&\alpha_{1,\out}^\out 	&1
\end{pmatrix}^{-1}\frac{1}{2}
\begin{pmatrix}
e^{\frac\pi 4i}&0&e^{-\frac\pi 4i}&0\\
0&e^{\frac\pi 4i}&0&e^{-\frac\pi 4i}\\
e^{-\frac\pi 4i}&0&e^{\frac\pi 4i}&0\\
0&e^{-\frac\pi 4i}&0&e^{\frac\pi 4i}
\end{pmatrix}
\begin{pmatrix}
1	&0	\\
0	&1	\\
0	&\alpha_2	\\
\alpha_1	&0	
\end{pmatrix},
\end{equation}
modulo $\ord(h^{\frac 2{2n+1}})$.
It follows from a direct computation of \eqref{matrixproduct} that
$$
A^\inc(E,h)=
\frac{e^{\frac\pi 4i}}2
\begin{pmatrix}
1 &-i(\alpha_2+\alpha_{2,\out}^\inc) -\alpha_{2,\inc}^\inc \\
-i(\alpha_1+\alpha_{1,\out}^\inc)-\alpha_{1,\inc}^\inc & 1 
\end{pmatrix}+\ord(h^{\frac 2{2n+1}}),
$$
$$
A^\out(E,h)=
\frac{e^{-\frac\pi 4i}}2
\begin{pmatrix}
1 &i(\alpha_2+\alpha_{2,\inc}^\out) -\alpha_{2,\out}^\out \\
i(\alpha_1+\alpha_{1,\inc}^\out)-\alpha_{1,\out}^\out & 1
\end{pmatrix}+\ord(h^{\frac 2{2n+1}}).
$$
Thus we obtain from the formula \eqref{T} that $T$ is given, modulo $\ord(h^{\frac 2{2n+1}})$, by
\be
\label{TEh}
-i\begin{pmatrix}
1 & \alpha_{2,\inc}^\inc-\alpha_{2,\out}^\out+i(2\alpha_2+\alpha_{2,\out}^\inc+\alpha_{2,\inc}^\out) \\
\alpha_{1,\inc}^\inc-\alpha_{1,\out}^\out+i(2\alpha_1+\alpha_{1,\out}^\inc+\alpha_{1,\inc}^\out) & 1
\end{pmatrix}.
\ee

Recall that the function $r_0(x)$ is assumed to have a compact support near the origin.
\begin{lemma}\label{lem:Coefficients-B1}
The entries $t_{12}$ and $t_{21}$ of $T$ have the following asymptotics as $h\to0^+$:
\be\label{eq:AiAiIntegrals}
\begin{aligned}
&t_{12}=-\frac{2\pi}{h^{1/3}}\int_{\R}\frac{r_0(x)}{\sqrt{\xi_1'\xi_2'}}\ai(h^{-2/3}\xi_1(x))\ai(h^{-2/3}\xi_2(x))dx+\ord(h^{\frac{2}{2n +1}}),\\
&t_{21}=-\frac{2\pi}{h^{1/3}}\int_{\R}\frac{r_0(x)}{\sqrt{\xi_1'\xi_2'}}\ai(h^{-2/3}\xi_1(x))\ai(h^{-2/3}\xi_2(x))dx+\ord(h^{\frac{2}{2n +1}}),
\end{aligned}
\ee
where $\xi_j(x)$ is the function $\xi(x)$ defined by \eqref{eq:analytic-eta} for $V=V_j$.
\end{lemma}

\begin{proof}
According to \eqref{TEh}, it is enough to compute $\alpha_j$ and $\alpha_{j,\dir'}^\dir$. 
For example we have
\begin{align*}
\alpha_1&= C_2(u_2^+,u_2^-;+\infty) [u_1^-]+\ord(h^{\frac{2}{2n+1}}) \\
&=\frac 1{h{\mathcal W}(u_2^+,u_2^-)}\int_0^\infty u_2^-(x)U_2u_1^-(x)dx+\ord(h^{\frac{2}{2n+1}}) \\
&=-\frac \pi{h^{\frac13}}\int_0^\infty \frac{r_0(x)}{\sqrt{\xi_1'(x)\xi_2'(x)}}\ai(h^{-2/3}\xi_1(x))\ai(h^{-2/3}\xi_2(x))dx+\ord(h^{\frac{2}{2n+1}}),
\end{align*}
where the last identity follows from the formula \eqref{uai} and the Wronskian formula before it. Similarly we have
\begin{align*}
&\alpha_2 =-\frac{\pi}{h^{\frac 13}}\int_0^\infty
\frac{r_0}{\sqrt{\xi_1'\xi_2'}}\ai(h^{-\frac 23}\xi_1)\ai(h^{-\frac 23}\xi_2)dx+\ord(h^{\frac 2{2n+1}}),\\
&\alpha_{j,\inc}^\out=\frac{i\pi}{2h^{\frac 13}}\int_{-\infty}^0
\frac{r_0}{\sqrt{\xi_1'\xi_2'}}\Ci^*(h^{-\frac 23}\xi_1)\Ci^*(h^{-\frac 23}\xi_2)dx+\ord(h^{\frac 2{2n+1}}),\\
&\alpha_{j,\out}^\inc=-\frac{i\pi}{2h^{\frac 13}}\int_{-\infty}^0
\frac{r_0}{\sqrt{\xi_1'\xi_2'}}\Ci(h^{-\frac 23}\xi_1)\Ci(h^{-\frac 23}\xi_2)dx+\ord(h^{\frac 2{2n+1}}), \\
&\alpha_{j,\out}^\out
=\frac{i\pi}{2h^{\frac 13}}\int_{-\infty}^0
\frac{r_0}{\sqrt{\xi_1'\xi_2'}}\Ci(h^{-2/3}\xi_j)\Ci^*(h^{-\frac 23}\xi_{\hat j})dx+\ord(h^{\frac 2{2n+1}}),\\
&\alpha_{j,\inc}^\inc
=-\frac{i\pi}{2h^{\frac 13}}\int_{-\infty}^0
\frac{r_0}{\sqrt{\xi_1'\xi_2'}}\Ci(h^{-\frac 23}\xi_{\hat j})\Ci^*(h^{-\frac 23}\xi_j)dx+\ord(h^{\frac 2{2n+1}}),
\end{align*}
for $\{j,\hat j\}=\{1,2\}$. 
Then we obtain \eqref{eq:AiAiIntegrals} from the identity
$
e^{-i\pi/4}\Ci+e^{i\pi/4}\Ci^*=2\Ai.
$
\end{proof}

The above lemma reduces our problem to the integral
\be
\label{IEh}
\mc{I}=\mc{I}(E,h):=-\frac{2\pi}{h^{\frac 13}}\int_{\R}\frac{r_0(x)}{\sqrt{\xi_1'(x)\xi_2'(x)}}\ai(h^{-\frac 23}\xi_1(x))\ai(h^{-\frac 23}\xi_2(x))dx.
\ee

\begin{proposition}\label{prop:key-int}
Let $|\re E|\le Lh^{\frac2{2n+1}}$ with a positive $L$.
One has 
\be\label{eq:key-int1}
\mc{I}(E,h)=-\frac{2\pi r_0(0)}{\sqrt{v_0}}\left (\frac{h}{q_n}\right )^{\frac 1{2n+1}} A_n\left(-\left(\frac{q_n}h\right)^{\frac2{2n+1}}\frac E{v_0}\right)+\ord(h^{\frac{2-\epsilon}{2n+1}}),
\ee
with $v_0:=\sqrt{V_1'(0) V_2'(0)}$, $0\le\epsilon<1$ is given in Theorem~\ref{MAINTH} and the function $A_n$ is defined by \eqref{defAn}.
\end{proposition}


The function $\mathcal{I}$ is analytic with respect to $E$ in a small complex neighborhood of $E=0$ according to the properties of the functions $\xi_1, \xi_2$ (see Appendix \ref{Appendix B}). For simplicity, in the following we work with real energies $E\in [-L h^{\frac{2}{2n+1}}, L h^{\frac{2}{2n+1}}]$ with an arbitrary positive $L$, but the same estimates hold for complex $E\in \mathcal{R}$.

We use different methods to obtain the principal term of this integral for the cases $n=1$ and $n\ge2$. The case $n=1$ is similar to \cite{FMW1}, and we send it to Lemma~\ref{lem:AiAi-}. 
More precisely, we apply this lemma after the reduction
\ben
\mc{I}(E,h)=-\frac{2\pi r_0(0)+\ord(h^{1/3})}{(v_0h)^{\frac13}}\int_\R 
\prod_{j=1,2}\ai\left(h^{-\frac23}(V_j'(0))^{\frac13}(x-a_j(E))\right)dx,
\een
where $a_j(E)$ is the unique root of $V_j(x)=E$ near $x=0$. The above reduction follows from $\xi_j(x)=(V_j'(0))^{\frac13}(x-a_j(E))+\ord(|x-a_j(E)|^2)$,
\ben
\int_{\R}\frac{(1-\chi(x/\sqrt h))r_0(x)}{\sqrt{\xi_1'(x)\xi_2'(x)}}\prod_{j=1,2}\ai(h^{-\frac23}\xi_j(x))dx=\ord(h^\infty),
\een
and
\ben
\int_\R\left(1-\chi(x/\sqrt h)\right)\prod_{j=1,2}\ai\left(h^{-\frac23}(V_j'(0))^{\frac13}(x-a_j(E))\right)dx=\ord(h^{\frac23}),
\een
with $\chi\in C_0^\infty(\R;[0,1])$ satisfying $\chi(x)=1$ near $x=0$.

Let us consider the case with $n\ge2$. 
Proposition~\ref{prop:key-int} is shown by using the following two lemmas.
Let $C>0$ be an $h$-independent large constant, and $\chi_\nu$ a smooth function such that
\be\label{chi-nu}
\chi_\nu(y)=\left\{
\begin{aligned}
&1\qtext{for}
y\ge2Ch^{2\nu}\\
&0\qtext{for}
y\le Ch^{2\nu}.
\end{aligned}
\right.
\ee

\begin{lemma}\label{lem:red-wkb}
Let $a(E)=\min (a_1(E), a_2(E))$. For $0<\nu\le\frac 13$, one has
\be\label{eq:Int-pm}
\mc{I}(E,h)=-\sum_\pm\int_{-\infty }^{a(E)} \chi_\nu(a(E)-x)\sigma(x,E)e^{\pm i\phi(x,E)/h}dx+\ord(h^{\nu}),
\ee
where, $\phi$ and $\sigma$ are given by
\ben
\begin{aligned}
&\phi(x,E)=\int_{a_1(E)}^x\sqrt{E-V_1(t)}dt-\int_{a_2(E)}^x\sqrt{E-V_2(t)}dt,\\
&\sigma(x,E)=\frac12r_0(x){\left [(E-V_1(x))(E-V_2(x))\right ]^{-\frac 14}}.
\end{aligned}
\een
\end{lemma}

\begin{proof}

Note that the absolute value of the integrand of $\mc{I}$ (including the pre-factor $h^{-\frac13}$) is $\ord(h^{-1/3})$ in the $\ord(h^{\frac23})$-neighborhood of $x=a$, 
decays exponentially for $x>a+Ch^{\frac23}$, and is bounded from above by $C(a-x)^{-\frac12}$ for $Ch^{\frac23}\le a-x\le 2Ch^{2\nu}$. 
These facts imply the estimate
\be\label{eq:cutoff-estimate}
\frac{2\pi}{h^{\frac 13}}\int_{\R}(1-\chi_\nu(a-x))\frac{r_0(x)}{\sqrt{\xi_1'\xi_2'}}\ai(h^{-\frac 23}\xi_1(x))\ai(h^{-\frac 23}\xi_2(x))dx
=\ord(h^{\nu}).
\ee

On the support of $\chi_\nu(a-x)$, we have $-h^{-\frac23}\xi_j(x)\gg1$, and
hence using the asymptotic formula \eqref{ai-} of the Airy function near $-\infty$, we have
\ben
\Ai(h^{-\frac23}\xi_1)\Ai(h^{-\frac23}\xi_2)=\frac{h^{\frac13}}{\pi\sqrt[4]{\xi_1\xi_2}}\prod_{j=1,2}\left(\sin\left(\frac{2(-\xi_j)^{\frac32}}{3h}+\frac\pi4\right)+\ord\left(\frac{h}{(-\xi_j)^{\frac32}}\right)\right).
\een
This is of WKB form when we rewrite the sine functions with exponential functions, and the phase functions are of the form (see \eqref{eq:analytic-eta})
\ben
\pm\frac 23 (-\xi_1(x))^{3/2}\pm\frac 23(-\xi_2(x))^{3/2}=
\pm\int_{a_1}^x\sqrt{E-V_1(t)}dt\pm\int_{a_2}^x\sqrt{E-V_2(t)}dt.
\een
The derivative of the phase is 
$\pm(\sqrt{E-V_1(t)}-\sqrt{E-V_2(t)})$ or $\pm(\sqrt{E-V_1(t)}+\sqrt{E-V_2(t)})$ and the first one has a zero (critical point) at the crossing point $x=0$, while the second one does not. Hence it is enough to compute the first one.
Using the identity \eqref{xi} for the symbol, we have
\ben
\frac{2\pi}{h^{\frac 13}}\int_{\R}\frac{\chi_\nu(a-x)r_0}{\sqrt{\xi_1'\xi_2'}}\ai(h^{-2/3}\xi_1)\ai(h^{-2/3}\xi_2)dx
= - \sum_\pm\int_{-\infty}^{a} \chi_\nu(a-x)\sigma e^{\pm i\phi/h}dx+\ord(h^{\frac13}).
\een
The Lemma follows from this with \eqref{eq:cutoff-estimate}.
\end{proof}

The turning points $a_1(E)$, $a_2(E)$ and their difference behave as $E\to 0$ like
\be\label{eq:asy-aj}
a_j(E)=\frac{E}{V_j'(0)}+\ord(E^2),
\quad
a_2(E)-a_1(E)
=\frac{q_n}{\sqrt{v_0}}\left(\frac E{v_0}\right)^{n}+\ord(E^{n+1}),
\ee
where $q_n$ is the non-zero constant defined by \eqref{qn}.

\begin{lemma}\label{lem:key-asym1}
As $E\to 0$ and $a(E)-x\to 0^+$,  we have
\be
\label{phi'}
\phi'(x,E)=\frac{-q_n x^n+\ord\left(\left(\left|E\right|
+\left|a-x\right|\right)^{n+1}\right)}{\sqrt{a_1-x}+\sqrt{a_2-x}},
\ee
\be\label{eq:phi-wkb}
\phi(x,E)
=q_n\int_0^{\sqrt{a-x}}\left(a-\tau^2\right)^n d\tau+\ord\left(|E|^{\frac{3n}2}+(|E|+|a-x|)^{n+\frac32}+|E|^{2n}|a-x|^{-\frac12}\right).
\ee
\end{lemma}

\begin{proof}
The derivative $\phi'$ is given by
\ben
\phi'(x,E)=\sqrt{E-V_1(x)}-\sqrt{E-V_2(x)}=\frac{V_2(x)-V_1(x)}{\sqrt{E-V_1(x)}+\sqrt{E-V_2(x)}}.
\een
One has $V_2(x)-V_1(x)=\sqrt{v_0}q_nx^n+\ord(x^{n+1})$, $\sqrt{E-V_j(x)}=\sqrt{v_0}\sqrt{a_j-x} (1+\ord(|E|+a_j-x))$. 
Since $x=\ord (|E|+|a-x|)$, $a_j=a+\ord (E^n)$, 
we obtain  \eqref{phi'}.

Let us show \eqref{eq:phi-wkb}. In the case $a_2(E)\le a_1(E)$, we have
\ben
\phi(x,E)=-\int^{a_2(E)}_x\phi'(t,E)dt-\int^{a_1(E)}_{a_2(E)}\sqrt{E-V_1(t)}dt.
\een
The second term is $\ord(|E|^{\frac{3n}2})$ since $a_1$ is a simple zero of $E-V_1(t)$ and $a_1(E)-a_2(E)=\ord(E^n)$. For the first term, we have 
\begin{align*}
-\int^{a_2}_x \frac{-q_nt^n}{\sqrt{a_1-t}+\sqrt{a_2-t}}dt
&=\int^{a_2}_x\frac{q_nt^n}{2\sqrt{a_1-t}}\left(1+\frac{a_2-a_1}{2\sqrt{a_1-t}(\sqrt{a_1-t}+\sqrt{a_2-t})}\right)^{-1}dt\\
&=q_n\int_{0}^{\sqrt{a_1-x}}(a_1-\tau^2)^n d\tau +
\ord\left(|E|^{\frac{3n}2}+|E|^{2n}|a_1-x|^{-\frac12}\right).
\end{align*}
Here, we used the change of the variable $\tau=\sqrt{a_1-t}$ and 
$\int_0^{\sqrt{a_1-a_2}}(a_1-\tau^2)^n d\tau=
\ord(|E|^{\frac{3n}2}).$
Together with \eqref{phi'}, we obtain \eqref{eq:phi-wkb}.
\end{proof}

\begin{proof}[Proof of Proposition~\ref{prop:key-int}]
Let $\mu$, $\nu$ be two numbers satisfying $\frac1{2n+3}<\mu<\frac1{2n+1}<\nu\le\frac13$. Then on the support of $\chi_\nu(a-\cdot)\left(1-\chi_\mu(a-\cdot)\right)$, one has $Ch^{2\nu}\le a-x\le 2Ch^{2\mu}$, and, according to Lemma~\ref{lem:key-asym1},
\be\label{eq:esti-sigma}
\sqrt{a-x}\,\sigma(x)=\frac{r_0(0)}{\sqrt{v_0}}\left(1+\ord(|E|+|a-x|)\right)
=\frac{r_0(0)}{\sqrt{v_0}}+\ord(
h^{2\mu}),
\ee
\ben
\phi(x)
=q_n\int_0^{\sqrt{a-x}}\left(a-\tau^2\right)^n d\tau+\ord\left(h^{(2n+3)\mu}+h^{\frac{4n}{2n+1}-\nu}\right).
\een
Thus, with $\lambda=(h/q_n)^{1/{2n+1}}$ and $\dl=\min\left\{2\mu,(2n+3)\mu-1,\frac{4n}{2n+1}-\nu-1\right\}$, we have 
\begin{align*}
&\sum_{\pm}\int_{-\infty}^{a}\chi_\nu (a-x)\left(1-\chi_\mu(a-x)\right)\sigma(x)e^{\pm i\phi(x)/h}dx\\
&=\left(\frac{r_0(0)}{\sqrt{v_0}}+\ord(h^\dl)\right)\sum_{\pm}
\int_0^{+\infty}\left((1-\chi_\mu)\chi_\nu\right)(z^2)\exp\left(\pm \frac{iq_n}h \int_0^{z}\left(\frac E{v_0}-\tau^2\right)^n d\tau\right) dz\\
&=\lambda\left(\frac{r_0(0)}{\sqrt{v_0}}+\ord(h^\dl)\right)
\int_{\R}\left((1-\chi_\mu)\chi_\nu\right)\left((\lambda\zeta)^2\right)\exp\left( i \int_0^{\zeta}\left(\frac E{\lambda^2 v_0}-\tau^2\right)^n d\tau\right) d\zeta\\
&=\lambda\left(\frac{2\pi r_0(0)}{\sqrt{v_0}}A_n\left(-\frac E{\lambda^2 v_0}\right)+\ord\left(h^\dl+h^{2n\left(\frac{1}{2n+1}-\mu\right)}+h^{\nu-\frac1{2n+1}}\right)\right).
\end{align*}
We have changed the integral variables $z=\sqrt{a-x}$, $\zeta=z/\lambda$ and used \eqref{An-cos}. 
On the other hand, by an integration by parts, we have
\be\label{eq:away}
\left|\int_{-\infty}^a \chi_\mu(a-x)\sigma(x)e^{\pm i\phi(x)/h}dx\right|
\le
Ch\int_{-\infty}^a\left|\left(\frac{\chi_\mu(a-x)\sigma(x)}{\phi'(x)}\right)'\right|dx=
\ord(h^{1-2n\mu}).
\ee
To minimize the above errors and $\ord(h^\nu)$ of  Lemma~\ref{lem:red-wkb}, one should take $\nu=\frac 13$ for $n\ge3$, $\nu=\frac 3{10}$ for $n=2$, and $\mu=\frac{4n+1}{(2n+1)(4n+3)}$. This gives the error estimate of $\ord(h^{\frac{2-\epsilon}{2n+1}})$.
\end{proof}

We end this section with the following proposition which asserts that the space of microlocal solutions near $(0,0)$ is of dimension $2$ and generated by $(w_1^-, w_2^-)$. This implies the first part of Theorem~\ref{prop:crossing-turning} and completes the proof of this theorem. 
\begin{proposition}
For any microlocal solution $u$ near the crossing point $(0,0)$ with $\|w\|_{L^2}\le1$ there exists an exact solution $w=c_1w_1^-+c_2w_2^-$ such that
\ben
u\equiv w\qtext{near} (0,0).
\een
\end{proposition}

\begin{proof}
We prove this in the same way as in \cite[Proposition~3.5]{AFH1}. 
One can first construct a quasi-mode $w_{\ope{quasi}}$ such that $u\equiv w_{\ope{quasi}}$ microlocally near $(0,0)$. 
In fact, let $w^\flat$, $w^\sharp$ be exact solutions with $\|w^\flat\|_{L^2}\le1$, $\|w^\sharp\|_{L^2}\le1$ such that, for $\dir=\inc,\out$,
\ben
u\equiv w^\dir\qtext{near}\gamma_1^\dir\cup\gamma_2^\dir.
\een
Then $w_{\ope{quasi}}$ is given by
\ben
w_{\ope{quasi}}=\chi(hD_x) (\psi u)+\chi^\inc(hD_x)(\psi w^\inc)+\chi^\out(hD_x)(\psi w^\out),
\een
where $\psi\in C_0^\infty(\R;[0,1])$ with $\psi=1$ near 0, and
$1=\chi+\chi^\inc+\chi^\out$ is a partition of unity  with $\chi=1$ near $0$ and $\chi^\inc$ (resp. $\chi^\out$) is supported only on the positive (resp. negative) real line.

The argument in the proof of \cite[Proposition~3.5]{AFH1} using Cramer's rule gives the existence of an exact solution $w$ such that  $w_{\ope{quasi}}=w+\ord(h^\infty)$ near $x=0$. 
\end{proof}

\section{Asymptotics in a smaller range}\label{sec:small}
In this section, we discuss Remark~\ref{rem1}.
Here an asymptotic formula with a constant leading coefficient and better remainder estimate  is given for $n\ge 2$ by restricting the real part of $E$ to a smaller zone with $\re E=o(h^{4/(2n+1)})$. 

In the case $\re E=o(h^{4/(2n+1)})$, we prove, instead of Formula~\eqref{eq:key-int1} of Proposition~\ref{prop:key-int},
\be\label{key-int2}
\mc{I}(E,h)=-\frac{2\pi r_0(0)}{\sqrt{v_0}}\left (\frac{h}{q_n}\right )^{\frac 1{2n+1}} A_n(0)
+\ord(h^{\frac2{2n+1}}+h^{\frac13}).
\ee
From this with \eqref{an0}, we deduce
$$
T(E,h)=-i{\rm Id}-\kappa_n(0) h^{\frac 1{2n+1}}\begin{pmatrix}
0& 1\\
1&0
\end{pmatrix}+\ord\left(h^{\frac2{2n+1}}+h^{\frac13}\right),
$$
instead of Formula~\eqref{eq:Thm2} in Theorem~\ref{prop:crossing-turning}. Then the asymptotic formulae in Remark~\ref{rem1} follows.

Let us prove \eqref{key-int2}.  There exists a smooth positive function $f=f(y,E)$ defined near $y=0$ such that $f(0,E)=1+\ord(E)$ and
\be\label{eq:phi-forEsmall}
\phi(x,E)=\frac{(-1)^{n}q_n f(a-x,E)}{2n+1}(a-x)^{n+\frac12}+\ord\left(\frac{\left|E\right|\left|a-x\right|^{n}+\left|E\right|^{2n}}{\left|a-x\right|^{\frac12}}
+\left|E\right|^{\min\left\{\frac{3n}2,\,n+\frac32\right\}}\right)
\ee
provided that $E=o(a-x)$. 
In fact, $a=\ord(E)$ gives
\ben
\int_0^{\sqrt{a-x}}(a-\tau^2)^n d\tau=
\frac{(-1)^n}{2n+1}(a-x)^{n+1/2}+ \ord(E(a-x)^{n-1/2}).
\een
The error estimated as $|a-x|^{n+\frac32}$ in Formula~\eqref{eq:phi-wkb} is included in $f(a-x,E)$. 

On the support of $\chi_\nu(a-\cdot)\left(1-\chi_\mu(a-\cdot)\right)$ with $0<\mu< \nu\le\min\{2/(2n+1),1/3\}$, one has $|E|\ll Ch^{2\nu}\le a-x\le 2Ch^{2\mu}$, and \eqref{eq:esti-sigma}. We moreover have
\ben
\phi(x)=\frac{(-1)^{n}q_n f(a-x,E)}{2n+1}(a-x)^{n+1/2}+
o(h^\dl),
\een
with $\dl=\min\left\{\frac4{2n+1}+(2n-1)\mu,\frac{6n}{2n+1},\frac{4n+6}{2n+1},\frac{4n}{2n+1}-\nu\right\}$, 
since $\phi$ admits \eqref{eq:phi-forEsmall}.
The same computation as before with the change of the variable $\lambda\zeta=\sqrt{a-x}\,f(a-x;E)^{1/(2n+1)}$ implies
\begin{align*}
&\sum_{\pm}\int_{-\infty}^{a}\chi_\nu (a-x)\left(1-\chi_\mu(a-x)\right)\sigma(x)e^{\pm i\phi(x)/h}dx\\
&=\lambda\left(\frac{r_0(0)}{\sqrt{v_0}}+\ord(h^{\dl-1})\right)
\int_{\R}\left((1-\chi_\mu)\chi_\nu\right)\left(a-x(\zeta)\right)\left(1+\ord((\lambda\zeta)^2)\right)\exp\left( \frac{i(-1)^n\zeta^{2n+1}}{2n+1}\right) d\zeta\\
&=\lambda\left(\frac{2\pi r_0(0)}{\sqrt{v_0}}A_n\left(0\right)+\ord\left(h^{\dl-1}+h^{2n\left(\frac{1}{2n+1}-\mu\right)}+h^{\nu-\frac1{2n+1}}\right)\right).
\end{align*}
To minimize these errors, one first need to set $\nu=\min\{1/3,2/(2n+1)\}$ (note that $\nu=1/3$ only when $n=2$).  Then $\mu=\frac{2n-1}{2n(2n+1)}$ is enough to estimate others as $\ord(h^\nu)$.

\appendix
\section{Airy functions}
\label{airyapp}
We review some known facts on the usual Airy functions and define a generalized Airy function used to describe our main result.
 
The  Airy functions are defined by the integrals 
\begin{align}
\label{aint}
\Ai(y)&:=\frac1{2\pi}\int_{\R}\exp \left ( i\left (\frac{\eta^3}3+y\eta\right )\right )d\eta,\\
{\rm Bi}(y)&:=\frac 1\pi\int_0^\infty\left [\exp \left (-\frac{\eta^3}3+y\eta\right )+\sin\left (\frac{\eta^3}3+y\eta\right )\right ]d\eta.
\end{align}
They are solutions to the Airy differential equation:
\be\label{eq:Airy1}
-v''(y)+yv(y)=0.
\ee
They are characterized by their asymptotic behavior as $y\to -\infty$:
\begin{equation}
\label{ai-}
{\rm Ai}(y)=\frac 1{\sqrt\pi}(-y)^{-\frac 14}\sin\left (\frac 23(-y)^{\frac 32}+\frac \pi 4\right )+\ord (|y|^{-\frac 74})\quad\text{as }y\to -\infty,
\end{equation}
\begin{equation}
\label{bi-}
{\rm Bi}(y)=\frac {-1}{\sqrt\pi}(-y)^{-\frac 14}\sin\left (\frac 23(-y)^{\frac 32}-\frac \pi 4\right )+\ord (|y|^{-\frac 74})\quad\text{as }y\to -\infty.
\end{equation}
As $y\to +\infty$, ${\rm Ai}$ is exponentially decaying whereas ${\rm Bi}$ is exponentially growing:
\begin{equation}
\label{ai+}
{\rm Ai}(y)=\frac 1{2\sqrt\pi}y^{-\frac 14}e^{-\frac 23 y^{\frac 32}}(1+\ord (y^{-\frac 32}))\quad\text{as }y\to +\infty,
\end{equation}
\begin{equation}
\label{bi+}
{\rm Bi}(y)=\frac 1{\sqrt\pi}y^{-\frac 14}e^{\frac 23 y^{\frac 32}}(1+\ord (y^{-\frac 32}))\quad\text{as }y\to +\infty.
\end{equation}
We introduce another linearly independent pair  $(\Ci,\Ci^*)$ defined by 
\begin{equation}
\label{ci}
\Ci:
=e^{i\frac \pi 4}\Ai+e^{-i\frac{\pi}{4}}\Bi,\quad
\Ci^*=e^{-i\frac \pi4}\Ai+e^{i\frac \pi4}\Bi.
\end{equation}
Then $\Ci^*$ is the complex conjugate of $\Ci$, and we have the following asymptotic formula 
\begin{equation}\label{ci-}
\Ci(y)=\frac{1}{\sqrt{\pi}}(-y)^{-\frac 14} e^{i{\frac23}(-y)^{\frac 32}}(1+\ord (y^{-\frac 32}))\quad\text{as }y\to -\infty.
\end{equation}
We call ${\rm Ci}$ an {\it outgoing} solution and ${\rm Ci}^*$ an {\it incoming} solution.

\begin{lemma}\label{lem:AiAi-}
For real numbers $\lambda_1,\lambda_2,\mu_1,\mu_2$ with $\lambda_1\neq\lambda_2  \neq 0$, we have
\be\label{eq:AiAi-}
\int_\R \Ai(\lambda_1(x-\mu_1))\Ai(\lambda_2(x-\mu_2))dx=
\frac{1}{(\lambda_2^3-\lambda_1^3)^{1/3}}\Ai\left(\frac{\mu_2-\mu_1}{\left(\lambda_1^{-3}-\lambda_2^{-3}\right)^{1/3}}\right).
\ee
Here we define $r^{1/3}=\ope{sgn}(r)\left|r\right|^{1/3}$ for $r\in\R$.
\end{lemma}

\begin{proof}
Recall the definition \eqref{aint} of the Airy function, which is a Fourier transform of the exponential function of a cubic function, more precisely, $\Ai(y)=(2\pi)^{-\frac12} \mathcal{F}(e^{-i \frac{\bullet^3}{3}})(y) = (2\pi)^{-\frac12}\mathcal{F}^{-1}(e^{i \frac{\bullet^3}{3}})(y)$. The Fourier transform and its inverse are defined by 
$$
\mathcal{F}u(y) := \frac{1}{\sqrt{2\pi}}\int_\mathbb R e^{-ixy} u(x)dx, \quad  \mathcal{F}^{-1}u(y) := \frac{1}{\sqrt{2\pi}}\int_\mathbb R e^{ixy} u(x)dx, \quad u\in L^1(\mathbb R).
$$
Then we have
\ben
\Ai(\lambda_j(x-\mu_j))=\frac1{\lambda_j\sqrt{2\pi}}\mathcal{F}(e^{i\mu_j\bullet}e^{-\frac i3(\frac{\bullet}{\lambda_j})^3}) (x)
=\frac1{\lambda_j\sqrt{2\pi}}\mathcal{F}^{-1}(e^{-i\mu_j\bullet}e^{\frac i3(\frac{\bullet}{\lambda_j})^3})(x) .
\een
We then apply the standard formula 
\ben
\int_\R (\mc{F}f)(x)(\mc{F}^{-1}g)(x)dx
=\int_\R f(y)g(y)dy
\een
of Fourier transformation to the integral \eqref{eq:AiAi-}, and obtain 
\begin{align*}
\int_\R \Ai(\lambda_1(x-\mu_1))\Ai(\lambda_2(x-\mu_2))dx&=
\frac1{2\pi \lambda_1\lambda_2}\int_\R e^{-i(\mu_2-\mu_1)y}e^{-\frac {iy^3}3 (\lambda_1^{-3}-\lambda_2^{-3})}dy \\
&= \frac{1}{\sqrt{2\pi}(\lambda_2^3-\lambda_1^3)^{1/3}} \mathcal{F}(e^{-\frac {i\bullet^3}3} ) \left(\frac{\mu_2-\mu_1}{\left(\lambda_1^{-3}-\lambda_2^{-3}\right)^{1/3}} \right)
\end{align*}


\end{proof}

In our study of microlocal connection formula, we encountered a  function defined by
$$
A_n(y)
:=\frac1{2\pi}\int_{\R}\exp\left(i\int_0^\eta(y+\tau^2)^n d\tau\right)d\eta,
$$
with positive integer $n$. This function can be seen as a generalization of the Airy function since
$$
A_1(y)={\rm Ai}(y).
$$
It has following properties: 
\begin{proposition}
One has
\be\label{An-cos}
A_n(y)=\frac1\pi\int_0^{+\infty}\cos\left(\int_0^\eta(y+\tau^2)^n d\tau\right)d\eta
=\frac1{2\pi}\int_\R\exp\left(i\int_0^\eta(-y-\tau^2)^n d\tau\right)d\eta,
\ee
and in particular $A_n(y)$ is real for real $y$.
At the origin, it has the value
\be
\label{an0}
A_n(0)=\frac1\pi (2n+1)^{\frac1{2n+1}}\bm{\Gamma}\left(\frac{2n+ 2}{2n+1}\right)\cos\left(\frac\pi{2(2n+1)}\right),
\ee
where $\bm{\Gamma}$ is the usual Gamma function defined by 
\begin{equation}\label{Gf}
\bm{\Gamma}(z) := \int_0^{+\infty} t^{z-1}e^{-t} dt, \quad {\rm Re}(z)>0.
\end{equation}
At infinity, it has the following asymptotic formulae:
\begin{align*}
A_n(y)&=a_n(y)|y|^{-\frac {n}{2(n+1)}}+\ord(|y|^{-\frac {3n+1}{2(n+1)}})\quad \text{as }y\to -\infty, \\
A_n( y )&= \ord( y ^{-\infty}) \quad \text{as } y \to +\infty, 
\end{align*}
where
$$
a_n( y )=\frac 1{\pi}\left(\frac{n+1}{2^n}\right)^{\frac 1{n+1}}\bm{\Gamma}\left(\frac {n+2}{n+1}\right)\times
\left\{
\begin{array}{ll}
\cos\left (
c_n |y| ^{n+\frac 12}-\frac \pi{2(n+1)}\right )\quad &(n:{\rm odd}), \\[10pt]
\cos\frac \pi{2(n+1)}\cos
(c_n |y|^{n+\frac 12})\quad &(n:{\rm even}),
\end{array}
\right.
$$
$$
c_n:=\int_0^1(1-t^2)^ndt=\frac 12 B(n+1,\frac 12)=\frac {2^{2n}(n!)^2}{(2n+1)!}.
$$
\end{proposition}
\begin{proof}
Formula~\eqref{An-cos} follows from the fact that $\int_0^\eta(y+\tau^2)^n d\tau$ is an odd function in $\eta$, and Formula~\eqref{an0} from 
the identity (for $a\in\R\setminus \{0\}$)
\be\label{deg-st-ph}
\int_\R e^
{iaz^{k+1}}dz=\frac{2 \bm{\Gamma}(\frac {k+2}{k+1})}{|a|^{\frac 1{k+1}}}
\times
\left\{
\begin{array}{ll}
\exp\left (\frac{\pi i}{2(k+1)}{\rm sgn}\, a\right ) \quad &(k:{\rm odd}),\\[8pt]
\cos\left (\frac{\pi }{2(k+1)}\right )\quad &(k:{\rm even}),
\end{array}
\right.
\ee
with $k=2n$. 
To show the asymptotic formulae, we rewrite the integral as
$$
A_n( y )=\frac{\sqrt y }{2\pi}\int_\R
\exp\left (i y ^{n+\frac 12}\int_0^x(1+t^2)^ndt\right )dx
$$
when $ y >0$ and
$$
A_n( y )=\frac{\sqrt{| y |}}{2\pi}\int_\R
\exp\left (i| y |^{n+\frac 12}\int_0^x(1-t^2)^ndt\right)dx
$$
when $ y <0$. Then the asymptotics as $ y \to+\infty$ follows from the fact that there is no stationary point for the phase $\int_0^x(1+t^2)^ndt$.

Let $\psi(x) :=\int_0^x(1-t^2)^ndt$ and $h= |y| ^{-(n+\frac 12)}$ for $ y <0$. Then
$$
A_n( y )=\frac{h^{-\frac 1{2n+1}}}{2\pi}\int_\R
e^{i\psi(x)/h}dx
$$
The critical points of the phase function $\psi$ are $x=\pm 1$, and near each point, it behaves like
\begin{align*}
\psi(x)&=c_n-\frac{2^n}{n+1}(1-x)^{n+1}+\ord((x-1)^{n+2})\,\,\text{ near }x=1, \\
\psi(x)&=-c_n+\frac{2^n}{n+1}(1+x)^{n+1}+\ord((x+1)^{n+2})\,\,\text{ near }x=-1.
\end{align*}
Then we have
$$
\int_\R
e^{i\psi(x)/h}dx=e^{\frac{i}hc_n}\int_\R\exp\left (
-\frac ih\frac {2^nz^{n+1}}{n+1}
\right )dz+e^{-\frac{i}hc_n}\int_\R\exp\left (\frac ih\frac {2^nz^{n+1}}{n+1}\right )dz+\ord(h^{\frac 2{n+1}}).
$$
The asymptotics of $A_n( y )$ as $ y \to -\infty$ follows  from \eqref{deg-st-ph}. 
\end{proof}

\section{Solutions to a scalar Schr\"odinger equation near a simple turning point}\label{Appendix B}
We define local solutions to the scalar Schr\"odinger equation
\be\label{eq:scalar-Q}
P_ju:=\left(-h^2\frac{d^2}{dx^2}+V_j(x)\right)u=Eu,
\ee
for complex $E$ near $0$.  We omit the index $j$ in this section.
Recall from Assumption \ref{H3} that $V$ extends to a holomorphic function in a small complex neighborhood of $x=0$, and $V(0)=0$ and $V'(0)>0$. 

It is known (see e.g., \cite{Ol,Ya,FMW1}) that there exists a pair of solutions  $(u^\out,u^\inc)$ to \eqref{eq:scalar-Q} which satisfies $\W(u^\out,u^\inc)=2ih^{-1}$ and
\be\label{eq:out-inc-Airy-sol}
\begin{aligned}
&u^\out(x)=\sqrt{\pi}h^{-1/6}\xi'(x)^{-1/2}\Ci(h^{-2/3}\xi(x))(1+\ord(h))\\
&u^\inc(x)=\sqrt{\pi}h^{-1/6}\xi'(x)^{-1/2}\Ci^*(h^{-2/3}\xi(x))(1+\ord(h))
\end{aligned}
\ee
as $h\to 0^+$ uniformly in a neighborhood of $x=0$. 
Here the function $\xi(x)$, depending on $E$, is initially defined for real $E$ near $0$ by
\be\label{eq:analytic-eta}
\xi(x)=\left\{
\begin{aligned}
&-\left(\frac 32 \int_x^{a(E)} \sqrt{E-V(t)}\,dt\right)^{2/3}\qtext{for}x\le a(E),\\
&\left(\frac 32 \int_{a(E)}^x \sqrt{V(t)-E}\,dt\right)^{2/3}\qtext{for}x> a(E),
\end{aligned}\right.\quad
\ee
where $a(E)$ is the unique root of $V(x)=E$ near $x=0$. In particular, $a(E) = \frac{E}{V'(0)} + \mathcal{O}(E^2)$ depends analytically on $E$ near $E=0$. Then, $\xi$ can be extended to an analytic function of $(x,E)$ in a small complex neighborhood of $(x,E)=(0,0)$. We can check that this function satisfies the following identity:
\be
\label{xi}
\xi(x)\xi'(x)^2=V(x)-E.
\ee

There exists $t=t(h)$  such that 
\ben
t=e^{-i\pi/4}+\ord(h)\quad \text{as }h\to 0^+, 
\een
and that the pair of solutions $(u^-,u^+)$ defined by
\be\label{eq:tj-scalar-pos}
(u^-,u^+):=(u^\out,u^\inc)
\left[\frac 12
\begin{pmatrix}
t&\bar{t}\\\bar{t}&t
\end{pmatrix}
\right]
\ee
satisfies $\W(u^-,u^+)=h^{-1}$ and 
\begin{align}
\label{uai}
&u^-=\sqrt{\pi}h^{-1/6}\xi'(x)^{-1/2}\Ai(h^{-2/3}\xi(x))(1+\ord(h)),\\
&u^+=\sqrt{\pi}h^{-1/6}\xi'(x)^{-1/2}\Bi(h^{-2/3}\xi(x))(1+\ord(h)),
\end{align}
as $h\to 0^+$ uniformly in a neighborhood of $x=0$.

Away from the turning point, the above solutions behave like WKB solutions :
\be
\begin{aligned}
&u^-=\frac{1}{2\sqrt[4]{V(x)-E}} \exp\left(-\frac 1h \int_{a(E)}^x\sqrt{V(t)-E}\,dt\right) 
\left(1+\ord\left(\frac{h^{2/3}}{x-a(E)}\right)\right),\\
&u^+=\frac{1}{\sqrt[4]{V(x)-E}} \exp\left(\frac 1h \int_{a(E)}^x\sqrt{V(t)-E}\,dt\right)\left(1+\ord\left(\frac{h^{2/3}}{x-a(E)}\right)\right),
\end{aligned}
\ee
for $x-a(E)\gg h^{2/3}$, and 
\be
\begin{aligned}
& u^-  =\frac{1}{\sqrt[4]{E-V(x)}} \exp\left(-\frac ih \int_{a(E)}^x\sqrt{E-V(t)}\,dt\right)
\left(1+\ord\left(\frac{h^{2/3}}{a(E)-x}\right)\right),\\
&  u^+  =\frac{1}{\sqrt[4]{E-V(x)}} \exp\left(\frac ih \int_{a(E)}^x\sqrt{E-V(t)}\,dt\right)
\left(1+\ord\left(\frac{h^{2/3}}{a(E)-x}\right)\right),
\end{aligned}
\ee
for $x-a(E)\ll -h^{2/3}$. 


\vspace{0.3cm}

\noindent
\textbf{Funding} 

\noindent 
The second author was supported by the JSPS KAKENHI Grant Number JP21K03303. The third author was supported by the Grant-in-Aid for JSPS Fellows Grant Number JP22J00430.

\vspace{0.3cm}

\noindent
\textbf{Disclosure statement} 

\noindent 
The authors report there are no competing interests to declare.

\end{document}